\documentclass[conference]{IEEEconf}
\IEEEoverridecommandlockouts
\usepackage[noadjust]{cite}
\usepackage{amsmath,amssymb,amsfonts}
\usepackage{algorithmic}
\usepackage{graphicx}
\usepackage{textcomp}
\usepackage{xcolor}
\usepackage{dsfont}
\usepackage{amsthm}
\usepackage{caption}

\newtheorem{theorem}{Theorem}[section]
\newtheorem{corollary}{Corollary}[section]
\newtheorem{lemma}[theorem]{Lemma}

\let \labelindent \relax
\usepackage{enumitem}
\usepackage{subcaption}
\usepackage[font=small]{caption}

\newcommand{\diag}{\text{diag}}

\usepackage{cite}
\usepackage{amsmath,amssymb,amsfonts}
\usepackage{algorithmic}
\usepackage{graphicx}
\usepackage{textcomp}
\usepackage{xcolor}
\def\BibTeX{{\rm B\kern-.05em{\sc i\kern-.025em b}\kern-.08em
    T\kern-.1667em\lower.7ex\hbox{E}\kern-.125emX}}
\begin{document}

\title{Active risk aversion in SIS epidemics on networks\\
\author{Anastasia~Bizyaeva*, Marcela~Ordorica~Arango*, Yunxiu~Zhou,
        Simon~Levin, and~Naomi~Ehrich~Leonard}
\thanks{This research was supported in part by ARO grant W911NF-18-1-0325 and AFOSR grant FA9550-24-1-0002.}
\thanks{* Signifies authors who contributed equally.}
\thanks{A. Bizyaeva, NSF AI Institute in Dynamic Systems and Dept. of Mechanical Engineering at the University of Washington; anabiz@uw.edu}
\thanks{M. Ordorica Arango, N.E. Leonard, Dept.
of Mechanical and Aerospace Engineering, Princeton University;  \{mo5674,naomi\}@princeton.edu.}
\thanks{Y. Zhou was with Dept. of Operations Research and Financial Engineering, Princeton University;
zhouyx1122@gmail.com
}
\thanks{S. Levin, Dept. of Ecology and Evolutionary Biology, Princeton University; slevin@princeton.edu}}

\maketitle

\begin{abstract} 
We present and analyze an actively controlled SIS (actSIS) model of interconnected populations to study how risk aversion strategies affect network epidemics. A population using a risk aversion strategy reduces its contact rate with other populations when it perceives an increase in infection risk. The network actSIS model relies on two distinct networks. One is a physical 
network that defines which populations come into contact with which others, thus how infection spreads. The other is a communication network, such as an online social network, that defines which populations observe the infection level of which others,  thus how information spreads. 
We prove that the system exhibits a transcritical bifurcation where an endemic equilibrium (EE) emerges. 
For regular graphs, we prove  that 
the endemic infection level is uniform across populations and reduced by the risk aversion strategy, relative to the network SIS endemic level. 
We show that when communication is sufficiently sparse, this initially stable EE loses stability in a secondary bifurcation. Simulations show that a new stable solution emerges with nonuniform infection levels. 


\end{abstract}

\section{Introduction}

Infectious diseases are a persistent public health challenge. Compartmental models like SI, SIS, SIR, and SIRI offer mathematically tractable frameworks for understanding the spread of infection in populations - see \cite{mei2017dynamics,zino2021analysis} for recent surveys. These models are useful in the design and evaluation of public health strategies for the containment and mitigation of infectious disease, such as the use of face masks in public, social distancing, partial lockdowns, and public vaccination strategies. The tractability of compartmental models makes them a compelling tool for analysis. However, they often do not account for 
human interactions and attitudes in the face of a disease, which 
can affect the spread of pathogens \cite{yang2022sociocultural}.
Various recent extensions to compartmental models on networks have been proposed to address these limitations. Typically, these extensions consider a second dynamic process that evolves alongside the infection spread, and in turn affects the infection spread dynamics \cite{chen2020networked}. 
Several recent works study epidemics coupled with game-theoretic updates to self-protective strategies within populations, including factors such as fatigue and various economic and social costs \cite{ye2021game,frieswijk2022mean,maitra2023sis,satapathi2023coupled,fenichel2011adaptive,morin2015social}. Other works 
couple epidemics and opinion dynamics \cite{xuan2020network,she2022networked}, examine the effects of population movements and the spatial dynamics of epidemics \cite{viboud2006synchrony,grenfell2001travelling}, and consider multi-layer networks \cite{pare2022multilayer,abhishek2023sis}.

In this paper, we add to this literature and study SIS epidemic dynamics for networked populations that update their inter-population contact rate in response to a dynamic assessment of infection risk based on their 
observations of infection level in other populations, e.g., from a social network.  We focus on populations of \textit{risk averters} that lower their contact rate with increased perceived infection risk.
The network SIS model of epidemics 
with state-dependent contact rates 
was recently studied in \cite{wang2021suppressing,wang2022state,walsh2023decentralised}. These works formulate distributed control laws to reduce endemic infection levels, with each population adjusting its contact rate in response only to \textit{its own} infection level. 
In contrast, we focus  on the impact of \textit{risk perception} on infection levels when populations can observe infection levels of \textit{other} populations. 

We propose and study \textit{network actSIS} (actively controlled Susceptible-Infected-Susceptible), a reactive model for networked populations.  The model relies on two networks: 1) a \textit{contact network} that defines which populations come into contact with which other populations and thus how infection spreads and 2) a \textit{communication network} that defines which populations observe the infection level of which other populations and thus how information spreads. 
We examine the effect of risk aversion  and the two network structures on epidemic spread. A notion of risk perception was studied in  epidemics in \cite{ye2021game,frieswijk2022mean}; however, the analysis uses 
a mean-field limit, which does not reveal the role of network structure.

Our main contributions are as follows. First, we extend actSIS, a well-mixed population model of actively controlled SIS epidemics, introduced in \cite{zhou2020active}, to network actSIS, a model of interconnected populations, with a contact network that governs infection spread and a communication network that governs risk assessment. Second, for this model we prove the emergence of an endemic equilibrium (EE) in a transcritical bifurcation. For regular contact and communication network graphs, we prove that risk aversion  lowers  the (uniform) infection levels at this 
EE. Third, we show that an EE in this model is not necessarily unique; it can lose stability in a secondary bifurcation if the communication network is sufficiently sparse.  Stability loss leads to emergence of an infection state that is heterogeneous even for regular  graphs, i.e., nonuniform infection levels, despite a high level of homogeneity in  model structure and parameters.

\section{Background\label{sec:background}}

\subsection{Notation and preliminaries}

Given a vector $\mathbf{x}\in \mathbb{R}^n$, $\diag(\mathbf{x})$ is the matrix with diagonal entries $x_i$, and zero off-diagonal entries. Given two vectors $\mathbf{x}, \mathbf{y} \in \mathbb{R}^n$, we say $\mathbf{x}\succ \mathbf{y}$ if $x_j > y_j$, and $x\succeq y$ if $x_j\geq y_j$  for all $j\in \{1,\ldots,n\}$, with similar element-wise definitions for matrices.
We denote by $\textbf{0}$ and $\textbf{1}$ the vectors of appropriate dimension whose entries are all equal to $0$ and $1$, respectively. A matrix is \textit{reducible} if it is similar to an upper-triangular matrix, and it is \textit{irreducible} if it is not reducible. A square matrix $A$ is \textit{Metzler} if 
$a_{ij}\geq 0$ for all $i\neq j$. Given a square matrix $A$, let $\{\lambda_i\}$ be the set of eigenvalues of $A$, the \textit{spectral radius} of $A$ is $\rho(A)=\max\{|\lambda_i| : i \in \{1,\ldots,n\}\}$. A \textit{leading eigenvalue} of $A$ is $\lambda_{max}(A) = \operatorname{argmax} \{ \operatorname{Re}(\lambda_i) \}$. 

$A$ is Hurwitz if $\operatorname{Re}(\lambda_{max}(A)) < 0$.
If $A$ is a Metzler, irreducible matrix, then $\rho(A) > 0$ is a simple leading eigenvalue of $A$ with associated left and right eigenvectors $\mathbf{w}\succ \mathbf{0}$, $\mathbf{v} \succ \mathbf{0}$. For nonnegative $A$ this is the classic Perron-Frobenius property and $\mathbf{w},\mathbf{v}$ are the Perron-Frobenius eigenvectors. The decomposition $A = T + U$ for a Metzler matrix $A$ is a \textit{regular splitting} if $U$ is Hurwitz and Metzler, and $T \succeq 0$. For a Metzler $A$ with regular splitting $T+U$, $\lambda_{max}(A) > 0 (= 0)$ if and only if $\rho(T U^{-1}) > 0 (= 0) $ \cite{pagliara2020adaptive}.

A \textit{graph} $\mathcal{G} = (\mathcal{V},\mathcal{E})$ is a collection of $n$ nodes $\mathcal{V}$ and edges $\mathcal{E}$. Node $j \in \mathcal{V}$ is a \textit{neighbor} of node $i \in \mathcal{V}$ on the graph $\mathcal{G}$ whenever the edge $(i,j) \in \mathcal{E}$. 
The adjacency matrix $A = (a_{ij})$ encodes relationships on graph $\mathcal{G}$, with $a_{ij} = 1$ whenever $(i,j) \in \mathcal{E}$ and $0$ otherwise. 
A graph is \textit{undirected} whenever $(i,j) \in \mathcal{E}$ if and only if $(j,i) \in \mathcal{E}$, i.e. $A = A^T$; it is \textit{directed} otherwise.
The in-degree of a node $j\in \mathcal{V}$ is $d_j = \sum_{k=1}^N a_{jk}$. A graph is \textit{regular} with in-degree $d$ if $d_j = d$ for all $j \in \mathcal{V}$.
A graph is \textit{strongly connected} if, for all pairs of vertices $i,j$, there is a path from $i$ to $j$. 
A graph $\mathcal{G}$ is strongly connected if and only if $A$ is irreducible. 
We  reference graphs by their associated adjacency matrix.

\subsection{SIS and actSIS models}

The Susceptible-Infected-Susceptible (SIS) model describes the spread of infection in a population of agents as
\begin{equation}
    \dot{p}=\beta (1-p)p-\delta p, \label{eq:SIS}
\end{equation}
where $p \in [0,1]$ is the proportion of the population that is currently infected, $\delta > 0$ is the recovery rate, and $\beta > 0 $ is the rate of infection transmission. The network SIS model generalizes the traditional SIS model by incorporating structured contact patterns among a set of populations through interactions over a contact network. 
Each node in the contact network graph $A$ signifies a separate population, and edges encode whether these populations come into contact. The infection level $p_j \in [0,1]$ in population $j$ evolves as
\begin{equation}
   \textstyle \dot p_j = \left(1-p_{j}\right) \beta \sum_{k=1}^{N} a_{j k} p_{k} -\delta p_{j}.\label{eq:NetworkSIS}
\end{equation}
$\beta$ and $\delta$ are infection and recovery rates. $a_{jk} = a_{kj} = 1$ if populations $j$ and $k$ comes into contact 
and $0$ otherwise.

While traditional SIS models consider 
$\beta$ to be a fixed parameter, in reality, populations can dynamically adapt their contact behavior, for example by practicing social distancing 
when infection levels rise. To account for this, the actively controlled SIS (actSIS) model 
\cite{zhou2020active} extends the traditional SIS framework by incorporating adaptive contact rates based on perceived infection risk. The model features two variables: the actual infection rate $p \in [0,1]$, and a filtered observation $p_s \in [0,1]$ that represents the \textit{perceived} infection level on which the population bases its risk assessment. The effective infection rate can be expressed as $\beta(p_s) = \bar\beta \alpha(p_s)$, where $\bar \beta$ is the intrinsic transmission rate of the disease and $\alpha: [0,1]\to [0,1]$ models the contact rate which actively adapts to perceived risk. The variables co-evolve according to 
\begin{equation}
    \dot{p}=\bar\beta \alpha(p_s)(1-p)p-\delta p, \ \ \ \  \tau_s \dot{p}_s=-p_s+p, \label{eq:actSIS_2d}
\end{equation}
where $\tau_s$ is a time constant that captures a potential delay in the transmission of information about infection levels.

\section{Network actSIS Model \label{sec:model}}
To study the effect of adaptive contact rates on the infection levels in interconnected populations, we extend the actSIS model from its original well-mixed statement \eqref{eq:actSIS_2d} to network actSIS, which incorporates network interactions. Analogously to the classic network SIS model, for the network actSIS model, we let $p_j\in[0,1]$ be the infection level of population $j$, and define a \textit{contact graph} with adjacency matrix $A = (a_{ij})$. We introduce $p_{sj} \in [0,1]$ as the perceived global infection level within population $j$, and we define a \textit{communication graph} with adjacency matrix $\hat{A} = (\hat{a}_{ij})$. The contact graph is associated with the spread of infection, whereas the communication graph encodes which populations influence the perception of risk of which others.  The two graphs can be the same, or may differ, for example with $A$ representing physical contact between populations and $\hat{A}$ representing interactions in an online social network.

For each population $j$, the effective contact rate with neighboring populations is $\beta_j(p_s) = \bar{\beta} \alpha_j(p_s)$ where $\Bar{\beta}$ is the intrinsic transmission rate and $\alpha_j:[0,1] \to [0,1]$ models the response to risk of population $j$. This formulation allows for potentially heterogeneous risk response strategies on the network. We assume each $\alpha_j$ is continuously differentiable on the unit interval (i.e. $C^1$), $|\alpha_j'(p_s)| < \infty$ for any $p_s \in [0,1]$, and $\alpha_j(0) \neq 0$ for all $j \in \mathcal{V}$. The infection level within population $j$ and its filtered observation of infection among its neighbors co-evolve according to 
\begin{align} 
  \textstyle\dot p_j&\textstyle= \left(1-p_{j}\right)\bar{\beta} \alpha_j(p_{sj}) \sum_{k=1}^{N} a_{j k} p_{k} -\delta_j p_{j},  \label{eq:pdot} \\
  \textstyle \tau_s \dot{ p}_{sj} & \textstyle = - p_{sj} + \frac{1}{\hat{d}_j} \left( \sum_{k=1}^N \hat{a}_{jk} p_k \right) \label{eq:phat}
\end{align}
where $\hat{d}_j =:\sum_{k=1}^N \hat{a}_{jk}$ is the in-degree of node $j$ on the communication graph  $\hat{A}$.

\section{Theoretical results for general model \label{sec:bif}}


In this section, we present an analysis of the network actSIS model \eqref{eq:pdot}, \eqref{eq:phat}. First, we establish that the model is well-posed, in the sense that the infection rate $p_j$ and perceived global infection rate $p_{sj}$ for each population $j \in \mathcal{V}$ remain within the interpretable set $[0,1] \subset \mathbb{R}$.  

\begin{theorem} [Well-Definedness]
The set $\mathcal{S} = [0,1]^{2N}$ is forward invariant for the network actSIS dynamics \eqref{eq:pdot},\eqref{eq:phat}. \label{thm:welldef}
\end{theorem}
\begin{proof}
To prove this statement we invoke Nagumo's theorem \cite[Theorem 4.7]{blanchini2008set}. The boundary of $\mathcal{S}$ is the set $\partial \mathcal{S}$ of all points at which $p_j = 0$, $p_j = 1$, $p_{sj} = 0$, and/or $p_{sj} = 1$ for one or more $j\in \mathcal{V}$.  Observe that whenever $(\boldsymbol{p},\boldsymbol{p}_s) \in \partial \mathcal{S}$, $p_j = 0 \implies \dot{p}_j = \bar{\beta}  \sum_{k =1}^N \alpha_{jk}(p_{sj}) p_k \geq 0$; $p_j = 1 \implies \dot{p}_{j} = - \delta_j <0$; $p_{sj} = 0 \implies \dot{p}_{sj} = \frac{1}{\tau_s\hat{d}_j} \sum_{k = 1}^N \hat{a}_{jk} p_k \geq 0$; and $p_{sj} = 1 \implies \dot{p}_{sj} = (-1 + \frac{1}{\hat{d}_j} \sum_{k = 1}^N \hat{a}_{jk} p_k )/\tau_{sj} \leq 0 $ for all $j \in \mathcal{V}$. Therefore, at every point on $\partial \mathcal{S}$, the vector field of the dynamics points along the tangent cone $\mathcal{T}_{\mathcal{S}}(\boldsymbol{p},\boldsymbol{p}_s)$. Compactness of $\mathcal{S}$ implies existence and uniqueness of solutions of \eqref{eq:pdot},\eqref{eq:phat}, in $\mathcal{S}$ for all $t \geq 0$ \cite[Theorem 3.3]{khalil2002nonlinear}, from which the statement follows by \cite[Corollary 4.8]{blanchini2008set}.
\end{proof}

Next, we  investigate the existence and stability of equilibria in the model. The Jacobian of the linearization of \eqref{eq:pdot},\eqref{eq:phat} about an arbitrary point $(\mathbf{p},\mathbf{p}_s) \in [0,1]^{2N}$ is 
\begin{equation}
    { \small J(\mathbf{p},\mathbf{p}_s) = \begin{pmatrix} - D +  \bar{\beta} F_1( P_1 A - P_2) & \bar\beta F_2 P_1 P_2 \\
    \frac{1}{\tau_s} \hat{\Delta}^{-1} \hat{A} & -\frac{1}{\tau_s} I\end{pmatrix} } \label{eq:networked_jac}
\end{equation}
where $D = \operatorname{diag}\{(\delta_1, \dots, \delta_N)\}$, $P_1 = \operatorname{diag}\{\mathbf{1} - \mathbf{p}\}$,  $P_2 = \operatorname{diag}\{ A \mathbf{p}\}$, $F_1 = \operatorname{diag}\{\alpha_1(p_{s1}), \dots, \alpha_N(p_{sN})\}$, $F_2 = \operatorname{diag}\{(\alpha_1'(p_{s1}), \dots, \alpha_N'(p_{sN}))\}$. Observe that the origin, i.e., infection-free state $(\mathbf{p},\mathbf{p}_s) = (\mathbf{0},\mathbf{0})$, is always an equilibrium of \eqref{eq:pdot}, \eqref{eq:phat}. At this equilibrium, \eqref{eq:networked_jac} simplifies to 
\begin{equation}
    { \small J_{IFE} = \begin{pmatrix}
        -D + \bar{\beta} \operatorname{diag}\{ \alpha_1(0),\dots,\alpha_N(0) \} A & 0 \\ 
        \frac{1}{\tau_s} \hat{\Delta}^{-1} \hat{A} & -\frac{1}{\tau_s} I
    \end{pmatrix}. } \label{eq:JIFE}
\end{equation}
In the following theorem, we establish a transcritical bifurcation of the origin in the network actSIS model, in which the infection-free equilibrium (IFE) loses stability, and a stable endemic equilibrium (EE) appears.

\begin{theorem}[Local bifurcation of EE] \label{thm:bif}
Consider \eqref{eq:pdot},\eqref{eq:phat}, assume $A$, $\hat{A}$ are irreducible, define $ \Tilde{A} = \operatorname{diag}\{(\alpha_1(0),\dots,\alpha_N(0)\}A$. 
Then the following  hold. 1) The IFE is locally exponentially stable for $\bar\beta < \bar\beta^* = \frac{1}{\rho( \Tilde{A} D^{-1})}$ and unstable for $\bar\beta > \bar\beta^*.$ 2) Let $\mathbf{v}^* \succ \mathbf{0}$, $\mathbf{w}^* \succ \mathbf{0}$ be right and left null eigenvectors of $-D + \bar{\beta}^* \Tilde{A}$, and suppose $K_1 = \sum_{i = 1}^N w_i^*\left( (\hat{\Delta}^{-1} \hat{A} \mathbf{v}^*)_i \alpha_i'(0)(A \mathbf{v}^*)_i - 2 v_i^* \alpha_i(0)(A\mathbf{v}^*)_i \right).$
If $K_1 \neq 0$, then at $\bar\beta = \bar\beta^*$ a branch of locally exponentially stable EE,  $\boldsymbol{p}^{*} \succ \mathbf{0}$, $\boldsymbol{p}_s^* = \hat{\Delta}^{-1} \hat{A} \boldsymbol{p}^* \succeq \mathbf{0}$, appear in a transcritical bifurcation along an invariant center manifold tangent to $\operatorname{span}(\bar{\mathbf{v}})$ at $(\mathbf{p},\mathbf{p}_s,\Bar{\beta}) = (\mathbf{0},\mathbf{0}, \Bar{\beta}^*)$, where $\bar{\mathbf{v}} = (\mathbf{v}^*, \hat{\Delta}^{-1} \hat{A} \mathbf{v}^*)$.  If $K_1 < 0$, the EE appear for $\bar\beta > \bar\beta^*$ and are locally exponentially stable; if $K_1 > 0$, the EE appear for $\bar\beta < 0$ and are unstable. 
\end{theorem}
\begin{proof}
1) Observe that $J_{IFE}$ \eqref{eq:JIFE} is block diagonal. The bottom block has eigenvalues $-1/\tau_s$ with multiplicity $N$ that are always negative, and the top block is the matrix $-D + \bar{\beta} \Tilde{A} $. It is an irreducible Metzler matrix by assumptions of the theorem, and therefore has a unique eigenvalue  $\lambda_{max}(-D + \bar{\beta} \Tilde{A})$ that is negative if and only if $\bar{\beta}\rho(\Tilde{A}D^{-1}) < 1$ and zero if and only if $\bar{\beta}\rho(\Tilde{A}D^{-1}) = 1$, with corresponding positive eigenvector $\mathbf{v}^*$. 
Therefore, when $\bar\beta < \bar\beta^* $, all of the eigenvalues of \eqref{eq:JIFE} are negative and the IFE is locally exponentially stable. When $\bar\beta  > \Bar{\beta}^*$, the IFE is unstable \cite[Theorem 4.7]{khalil2002nonlinear}. 2) When $\bar{\beta} = \bar{\beta}^*$, $J_{IFE}$ has a zero eigenvalue; existence of an attracting invariant center manifold of the IFE follows by the Center Manifold Theorem. 
To classify the bifurcation we utilize results from singularity theory of bifurcations \cite{Golubitsky1985} and compute a Lyapunov-Schmidt (LS) reduction expansion. A LS reduction is a low-dimensional algebraic equation that describes the local topology of the bifurcation diagram (its solution sets are in one-to-one correspondence with the zero level sets of \eqref{eq:pdot},\eqref{eq:phat} near $\bar{\beta}^*$). A second order LS reduction for \eqref{eq:pdot},\eqref{eq:phat} reads $ h(x,\bar{\beta}) = K_0 x (\bar\beta - \bar\beta^*) + \bar{\beta}^*  K_1  x^2 + h.o.t.$
where $x$ is a scalar coordinate along a projection of the dynamics onto the kernel of $J_{IFE}$ at  $\bar{\beta} = \bar{\beta}^*$ and $K_0 = \langle \mathbf{w}^*, \tilde{A} \mathbf{v}^* \rangle > 0$. The details of this calculation are included in the Appendix. To identify the transcritical bifurcation, we verify that the LS reduction coefficients satisfy the necessary and sufficient conditions of the recognition problem \cite[Proposition 9.3]{Golubitsky1985}:   $h(0,\bar\beta^*) = h_x(0,\bar\beta^*) = h_\beta(0,\bar\beta^*) = 0$ where $\beta = \bar \beta - \bar \beta^*$. Furthermore, $\operatorname{sign}\{h_{xx}(0)\} = \operatorname{sign}(K_1)$ and $\operatorname{sign}\{\det d^2 h(0) \} = -1$, where $\det d^2 h(0) = \begin{vmatrix} h_{xx}(0) & h_{x \beta}(0) \\ h_{x \beta}(0) & h_{\beta \beta}(0) \end{vmatrix} = \begin{vmatrix} \bar\beta^*K_1 & K_0 \\ K_0 & 0 \end{vmatrix} = - (K_0)^2.$
We conclude that the bifurcation is transcritical. 
Local stability of the bifurcating branch of the EE follows by \cite[ Chapter I, \S4, Theorem 4.1]{Golubitsky1985}. 
\end{proof}


\vspace{-4mm}

\begin{corollary} \label{cor:bif_specialized}
    Consider \eqref{eq:pdot},\eqref{eq:phat}. 
    1) If $\alpha_j'(0) \leq 0$ for all $j \in \mathcal{V}$, then $K_1 < 0$ and locally exponentially stable EE appear for $\bar\beta > \bar\beta^*$;
    2) If $\delta_j = \delta > 0$, $\alpha_{j}(0) = 1$ for all $j \in \mathcal{V}$, then $\bar{\beta}^* = \delta/  \lambda_{max}(A)$ and $\mathbf{v}^* \succ \mathbf{0}$ is a Perron-Frobenius eigenvector of $A$.
\end{corollary}
Case 1 of Corollary~\ref{cor:bif_specialized} is relevant in the next section where we study networks of risk averter populations. Case 2 shows mild conditions under which the critical value $\bar\beta*$ is directly related to the leading eigenvector of unweighted contact matrix $A$ and the EE infection levels are predicted by the elements of the Perron-Frobenius eigenvector of $A$. 

According to Theorem \ref{thm:bif}, the IFE in the network actSIS model \eqref{eq:pdot},\eqref{eq:phat} loses stability in a transcritical bifurcation. This result is an extension of a similar transcritical bifurcation of an EE that is well-known to exist in the SIS, network SIS, and more recently actSIS models \cite{mei2017dynamics}, \cite{zhou2020active}. Interestingly, the communication adjacency matrix $\hat{A}$ does not determine the bifurcation point of the EE and does not play a significant role in shaping the steady-state infection level at the EE near the bifurcation point. However, recall that in the network SIS model, the EE is unique and stable for all $\bar{\beta} > \Bar{\beta}^*$. We will show in the following section that this is not always the case for the EE of Theorem \ref{thm:bif}, and properties of this equilibrium at larger values of $\bar{\beta}$ depend strongly on properties of the communication graph $\hat{A}$.

\vspace{-2mm}
\section{Risk aversion in network actSIS epidemics \label{sec:riskaverters}}

In this section, we specialize the network actSIS model \eqref{eq:pdot},\eqref{eq:phat} to populations of \textit{risk averters}, i.e. populations that decrease contact rates with other populations when their perceived risk of infection rises, meaning $\alpha_j'(p_{sj}) < 0$ whenever $p_{sj}  \in (0,1)$. 
In shown simulations, all populations follow the saturating strategy $\alpha_A$ proposed in \cite{zhou2020active}
\begin{equation}
    {\tiny \alpha_j(p_{sj}) = \alpha_A(p_{sj}) := \frac{\mu^\nu (1 - p_{sj})^\nu}{p_{sj}^{\nu}(1 - \mu)^\nu + \mu^\nu (1 - p_{sj})^\nu} } \label{eq:alpha_A}
\end{equation}
where $\mu \in (0,1)$  describes the midpoint of saturation, and $\nu > 1$ tunes the steepness of the saturation slope. 
First we establish a lemma for the well-mixed actSIS model \eqref{eq:actSIS_2d}.

\begin{lemma}[Risk aversion lowers EE in well-mixed model] \label{lem:actSIS_averters}Consider the models from \eqref{eq:SIS} and \eqref{eq:actSIS_2d} with $\beta = \Bar{\beta}>\delta$ and $\alpha'(p)<0$ for all $p \in (0,1)$. Let $p^*_{\text{SIS}}(\bar\beta,\delta)$ and $p^*_{\text{actSIS}}(\bar\beta,\delta)$ be the steady-state infection level in the population at the EE of the corresponding model, with shared $\delta$, $\bar{\beta}>\bar\beta^*$. If $\alpha$ is a risk aversion strategy, then $ p^*_{\text{SIS}}(\bar\beta,\delta)\geq p^*_{\text{actSIS}}(\bar\beta,\delta)$.



\end{lemma}
\begin{proof}
Define the function $\gamma(p,s):=(1-s)+s\alpha(p)$
with $s\in [0,1]$, and consider \eqref{eq:actSIS_2d} with $\alpha$ replaced by $\gamma$. 
Notice that $s=0$ recovers the standard SIS model \eqref{eq:SIS}, and $s=1$ recovers the actSIS model \eqref{eq:actSIS_2d}. 
The equilibria of this modified model are defined implicitly as a function of $s$ by $\bar\beta\gamma(p^*,s)(1-p^*)p^*-\delta p^*=0$
and $p_s^* = p^*$. 
Differentiating with respect to $s$ and solving for $ \frac{dp^*}{ds}$,
\begin{equation}
    \frac{dp^*}{ds} = \frac{(1-p^*)(1-\alpha(p^*))}{s \big(\alpha'(p^*) (1-p^*) +  (1 - \alpha(p^*)) \big) - 1}. \label{eq:dpdt}
\end{equation}
For any $p^* \in (0,1)$, $\alpha(p^*) \in (0,1)$ and the numerator in \eqref{eq:dpdt} is positive. Observe that whenever 
$ \alpha'(p^*) < \frac{1/s - (1 - \alpha)}{1 - p^*}$
the denominator of \eqref{eq:dpdt} is negative and $\frac{dp^*}{ds}<0$.  Given any $s \in (0,1)$, $1 < 1/s < \infty$, which means that $\alpha'(p^*) < \frac{\alpha(p^*)}{1 - p^*}$ is sufficient to conclude $\frac{dp^*}{ds} < 0$. 
In turn, this condition is satisfied for all $p^* \in (0,1)$ given any $\alpha$ satisfying assumptions of the lemma, since $\alpha(p^*)/(1-p^*) > 0$, and  $\alpha'(p) < 0$ by definition of risk aversion.  
\end{proof}

\begin{figure}
    \centering
    \begin{subfigure}{.24\textwidth}
        \includegraphics[width=1.0\linewidth]{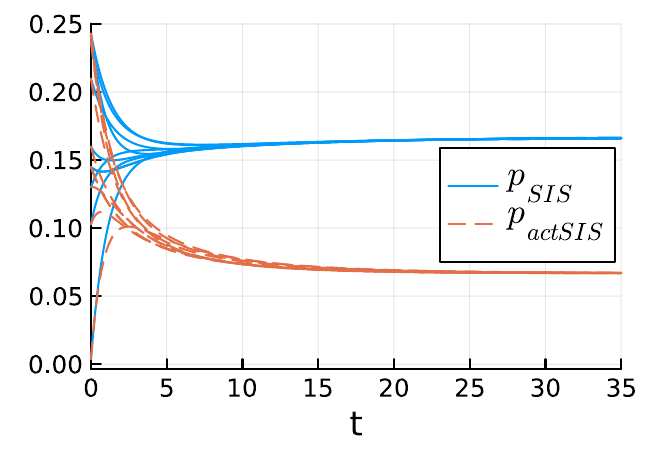}
        \caption{$d=2$ and $\hat{d}=5$.}
    \end{subfigure}%
    \begin{subfigure}{.24\textwidth}
        \includegraphics[width=1.0\linewidth]{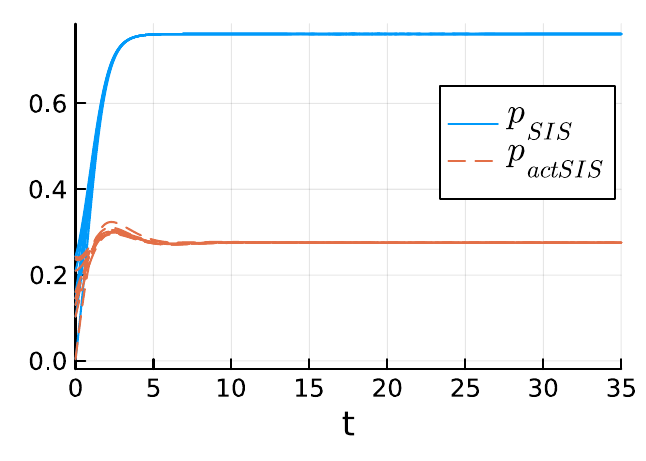}
        \caption{$d=7$ and $\hat{d}=5$.}
    \end{subfigure}
    \caption{Simulation for network actSIS \eqref{eq:pdot},\eqref{eq:phat}  with $10$ nodes and regular contact graph and communication graph, using \cite{rackauckas2017differentialequations}. The degree of $\hat{A}$ is constant but $A$ is sparse (left) or dense (right). 
    $\bar\beta = 0.3, \delta = 0.5, \mu = 0.5, \nu = 1.7, \tau_s = 1$.}
\label{fig:EEcomparison}
\end{figure}

Lemma \ref{lem:actSIS_averters} proves, in a different way than in \cite{zhou2020active}, that risk aversion lowers infection rate at EE in \eqref{eq:actSIS_2d}. Next, we consider risk aversion strategies in network actSIS, and prove an analogous result. We restrict our attention to contact and communication graphs that are undirected and regular. By Corollary \ref{cor:bif_specialized}, for populations of risk averters, the necessary and sufficient conditions for a transcritical bifurcation of a stable EE in Theorem \ref{thm:bif} are always satisfied. 

\begin{lemma}[Uniform Endemic Equilibrium (UEE)] \label{lem:UEE}
    Let $A$ and $\hat{A}$ correspond to connected regular graphs with degrees $d$ and $\hat{d}$, respectively, and assume $\delta_j = \delta$, $\alpha_j = \alpha$ for all $j \in \mathcal{V}$.
    1) The unique, globally stable EE of the network SIS model \eqref{eq:NetworkSIS} with contact graph $A$ is $\mathbf{p}^*_{SIS} = p_{SIS}^* \mathbf{1}$ where $ p_{SIS}^*  = 1 - \frac{\delta}{\Bar{\beta} d}$;
    2) The EE bifurcating from IFE for the network actSIS model \eqref{eq:pdot},\eqref{eq:phat} with contact graph $A$ and communication graph $\hat{A}$ is $(\mathbf{p}_{actSIS}^*,\mathbf{p}_{s}^*) = (p_{actSIS}^* \mathbf{1}, p_{actSIS}^* \mathbf{1})$ where $p_{actSIS}^*$ 
    satisfies $\bar\beta d (1- p_{actSIS}^*) \alpha(p_{actSIS}^*) - \delta = 0$. 
    This equilibrium exists for all $\bar\beta > \delta/d$.
\end{lemma}
Lemma~\ref{lem:UEE} says that the EE (that bifurcates from the IFE) for regular graphs corresponds to uniform levels of infection among the populations in network SIS and network actSIS. We omit the proof of Lemma \ref{lem:UEE} as it amounts to plugging in the stated expressions into the two sets of model equations. 


\begin{theorem}[Risk aversion lowers UEE on regular graphs] \label{thm:uniform_EE_lower}
    Let $A$ and $\hat{A}$ correspond to connected regular graphs with degrees $d$ and $\hat{d}$, respectively,  and assume $\delta_j = \delta$, $\alpha_j = \alpha$ for all $j \in \mathcal{V}$, and $\alpha'(p)<0$ for all $p \in (0,1)$. Let $\mathbf{p}^*_{SIS}$ be the UEE of \eqref{eq:NetworkSIS} with $\beta = \bar\beta>\delta/d$, and let $(\mathbf{p}_{actSIS}^*,\mathbf{p}_{s}^*)$ be the UEE of \eqref{eq:pdot},\eqref{eq:phat}. Then $\mathbf{p}_{SIS}^* \succeq \mathbf{p}^*_{actSIS}$. 
    

\end{theorem}
\begin{proof}
    Existence of UEE in both models for $\bar\beta \geq \delta/d$ follows from Lemma \ref{lem:UEE}. 
   To show  $p_{SIS}^* \geq p_{actSIS}^*$ 
    let $\gamma(p,s):=(1-s)+s\alpha(p)$ with $s\in [0,1]$. 
    The solution to $\bar\beta d \gamma(p^*,s)(1-p^*)-\delta =0$ recovers $p_{SIS}^*$ when $s = 0$ and $p_{actSIS}^*$ when $s = 1$. The rest of the proof follows as for Lemma \ref{lem:actSIS_averters}.
    \end{proof}

We illustrate Theorem \ref{thm:uniform_EE_lower} with numerical simulations in Fig. \ref{fig:EEcomparison}. 
As predicted,  the infection level at the UEE is uniform among populations and lower for the network actSIS dynamics than for the network SIS dynamics, with all other relevant parameters shared between the models.   %

\begin{theorem}[Stability of UEE on regular graphs]

Consider \eqref{eq:pdot},\eqref{eq:phat}. 
Let $A$ and $\hat{A}$ correspond to connected regular graphs with degrees $d$ and $\hat{d}$, respectively,  and assume $\delta_j = \delta$, $\alpha_j = \alpha$ for all $j \in \mathcal{V}$. 
Define
\begin{equation}
  g(p) \hspace{-1mm} = \hspace{-1mm} - \delta p + (1\hspace{-1mm}-\hspace{-1mm}p) \hspace{-1mm}\left(\hspace{-1mm}\lambda_{max}\hspace{-1mm}\left(\frac{1}{d} A + \frac{1}{\hat{d}} \frac{\alpha'(p)}{\alpha(p)} p \hat{A} \right) - 1 \right) \,. \label{eq:stab_cond}
\end{equation}
If $g(p) = 0$ for some $p \in (0,1)$,
then there exists a critical value $ \bar\beta_2 = \frac{\delta}{d (1 - p_2) \alpha(p_2)} > \bar\beta^*$ 
where $p_2$ is the smallest $p \in (0,1)$ for which $g(p) = 0$. If $p_2$ is not a unique solution to $g(p) = 0$ in $(0,1)$, the UEE $(\mathbf{p}^*, \mathbf{p}_s^*) = (p^* \mathbf{1}, p^* \mathbf{1})$ is locally exponentially stable whenever $\Bar{\beta}^* < \bar{\beta} < \bar{\beta}_2$ and unstable for $\bar{\beta} \in (\bar{\beta}_2, \Bar{\beta}_3)$ for some $\Bar{\beta}_3 = \frac{\delta}{d (1 - p_3) \alpha(p_3)}$, $p_3 > p_2$, $g(p_3) = 0$. If $p_2$ is a unique solution to $g(p) = 0$ in $(0,1)$, then the UEE is unstable for all $\bar\beta > \bar\beta_2$.    \label{thm:uniform_EE_stab}
\end{theorem}
\begin{proof}
Existence of the UEE for $\bar{\beta} > \bar{\beta}^*$ was established in Theorem \ref{thm:uniform_EE_lower}, and its local exponential stability at the onset of the transcritical bifurcation follows from Theorem \ref{thm:bif}. From the form of the matrix \eqref{eq:networked_jac} we infer that its eigenvalues are continuous in the parameter $\Bar{\beta}$, which means that the UEE is locally exponentially stable until one or more eigenvalues of $J(\mathbf{p}^*,\mathbf{p}_s^*)$ are singular for some $\bar{\beta} > \bar{\beta}^*$. At $(p^* \mathbf{1}, p^* \mathbf{1})$, \eqref{eq:networked_jac} simplifies to 
\begin{equation}
    { \tiny J_{U}(p^*) = \begin{pmatrix} - \delta \left(1 + \bar\beta d p^* \alpha(p^*) \right) I + \frac{\delta}{d} A  & \delta \frac{\alpha'(p^*)}{\alpha(p^*)} p^* I \\
    \frac{1}{\tau_s} \frac{1}{\hat{d}} \hat{A} & - \frac{1}{\tau_s}  I
    \end{pmatrix} }, \label{eq:JU}
\end{equation}
where we used $\bar{\beta} (1 - p^*) \alpha(p^*)d = \delta$ to eliminate terms.
By Schur's formula 
$\det(J_U(p^*) - \lambda I)  = \det\left(- \left(\frac{1}{\tau_s}+\lambda\right)I\right)\times \det\left(- (\delta \left(1 + \bar\beta d p^* \alpha(p^*) \right)+ \lambda) I + \frac{\delta}{d} A + \frac{\delta}{\hat{d}} \frac{\alpha'(p^*)}{\alpha(p^*)}p^*\hat{A}\right).$ 
Since $p^*$ is an implicit function of $\bar{\beta}$, we see that the eigenvalues of $J_U(p^*)$ that depend on $\bar{\beta}$ are eigenvalues of
\begin{equation}
    J_{eff}(p^*) \hspace{-1mm} =\hspace{-1mm} - \delta \left(1 \hspace{-0.5mm} + \hspace{-0.5mm} \bar\beta d p^* \alpha(p^*) \right) I \hspace{-0.3mm}+ \hspace{-0.3mm}\frac{\delta}{d} A \hspace{-0.3mm}+ \hspace{-0.3mm} \frac{\delta}{\hat{d}} \frac{\alpha'(p^*)}{\alpha(p^*)}p^*\hat{A}. \label{eq:effective_jacobian}
\end{equation}
Each eigenvalue of \eqref{eq:effective_jacobian} takes the form $\lambda_i(J_{eff}(p^*)) = - \delta \left(1 + \bar\beta d p^* \alpha(p^*) \right)  + \delta \lambda_i(M(p^*))$, where $M(p^*) = \frac{1}{d} A + \frac{1}{\hat{d}} \frac{\alpha'(p^*)}{\alpha(p^*)} p^* \hat{A}$. We are interested in the case that the largest eigenvalue $\lambda_{max}(J_{eff}(p^*)) = 0$, as it corresponds to the uniform equilibrium changing stability. At a singular point, $\bar{\beta} = \frac{1}{d p^* \alpha(p^*)}(-1 + \lambda_{max}(M(p^*))) = \frac{\delta}{(1-p^*) d \alpha(p^*)}$,
where the first expression is derived from the zero eigenvalue condition and the second is from the equilibrium condition. Algebraic manipulations of the above statement lead to the condition $g(p) = 0$. Finally, notice that $ \lambda_{max}(J_{eff}(p^*)) = \delta g(p^*)/(1-p^*)$, which means $\operatorname{sign}(\lambda_{max}(J_{eff}(p^*))) = \operatorname{sign}(g(p^*))$. Local stability of the UEE can thus be inferred from the sign of $g(p^*)$ and the theorem follows. 
\end{proof}

According to Theorem \ref{thm:uniform_EE_stab}, the UEE in the network actSIS model can lose stability. We illustrate the relationship between this property and the degree of the communication graph in Fig. \ref{fig:g}. We plot \eqref{eq:stab_cond} for a fixed $A$ and six representative choices of $\hat{A}$. When $\hat{d}$ is large, $g(p)$ remains negative for all $p \in (0,1)$, which means by Theorem \ref{thm:uniform_EE_stab} that the UEE remains stable for all $\bar\beta> \bar\beta^*$. However, when $\hat{A}$ is sparse, $g(p)$ crosses zero and the origin loses stability. This property is relatively unaffected by the degree of the contact graph, as the curves are only slightly perturbed between a sparse and dense choice of contact graph. In Fig. \ref{fig:beta2} we compute the critical value $\bar\beta_2$ across many simulation trials with randomly generated $A$, $\hat{A}$. We observe that for the parameter range within which the UEE loses stability, the average value and variance of $\bar\beta_2$ 
increase with the communication degree $\hat{d}$. The exact properties of these relationships may change based on the choice of risk aversion strategy $\alpha$; we leave quantifying them to future work. 

In Fig. \ref{fig:EEbifurcation_beta} we contrast simulations of \eqref{eq:pdot},\eqref{eq:phat} with $\bar\beta^* < \bar\beta < \bar\beta_2$ (left) and with $\bar\beta > \bar\beta_2$ (right). Once the UEE loses stability, the network settles to a \textit{nonuniform state}. The average level of infection across the network remains close to its level at the UEE; however, different populations settle to different infection levels at steady-state. Interestingly, this heterogeneity emerges despite a high level of regularity in the parameters and network choice of the graph, and is often a direct consequence of the sparsity of communication. Due to this sparsity, some populations overestimate the average risk, and others underestimate it, which in turn translates to a nonuniform infection level across the network. 



    
    

\begin{figure}
    \centering
    \includegraphics[width = 0.48\textwidth]{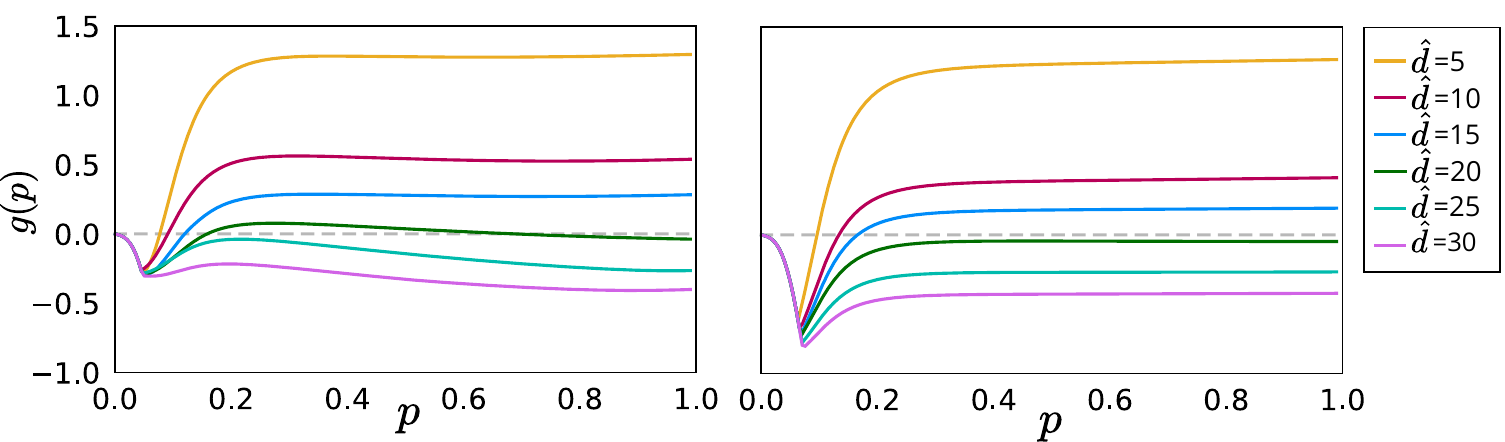}
    \caption{Plot of \eqref{eq:stab_cond} over 40 nodes with fixed sparse contact graph: $d = 5$ (left), and fixed dense contact graph: $d=30$ (right). 
    Each curve in each plot corresponds to an $\hat{A}$ with different $\hat d$. $\hat{A}$ is generated randomly using \cite{hagberg2008exploring}.
    $\delta = 1, \mu = 0.1, \nu=3$.}
    \label{fig:g}
\end{figure}


\begin{figure}
    \centering
    \begin{subfigure}{.24\textwidth}
        \includegraphics[width=1.0\linewidth]{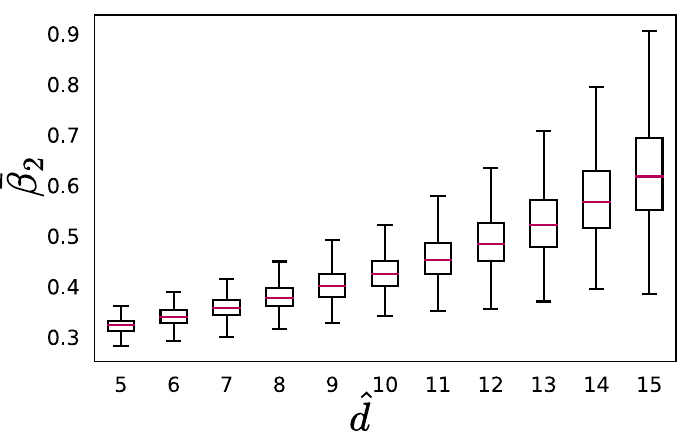}
        \caption{$d = 5$}
    \end{subfigure}%
    \begin{subfigure}{.24\textwidth}
        \includegraphics[width=1.0\linewidth]{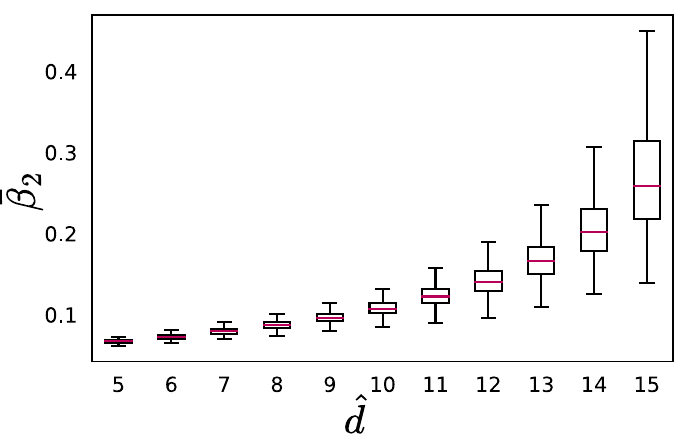}
        \caption{$d = 30$} 
    \end{subfigure}
    \caption{Distribution of the critical point $\bar{\beta}_2$ of Theorem \ref{thm:uniform_EE_stab} for 40 nodes with (a) sparse contact graph; (b) dense contact graph; and communication graphs with range of $d$. 
    Each box plot contains $\bar{\beta}_2$ from 1000 simulations. For each simulation, $A$ and $\hat{A}$ of appropriate $d$, $\hat{d}$ are generated randomly using \cite{hagberg2008exploring}. 
    $\delta = 1, \mu = 0.1, \nu=3$.
    \label{fig:beta2} }
\end{figure}


\begin{figure}
    \centering
    \begin{subfigure}{.24\textwidth}
        \includegraphics[width=1.0\linewidth]{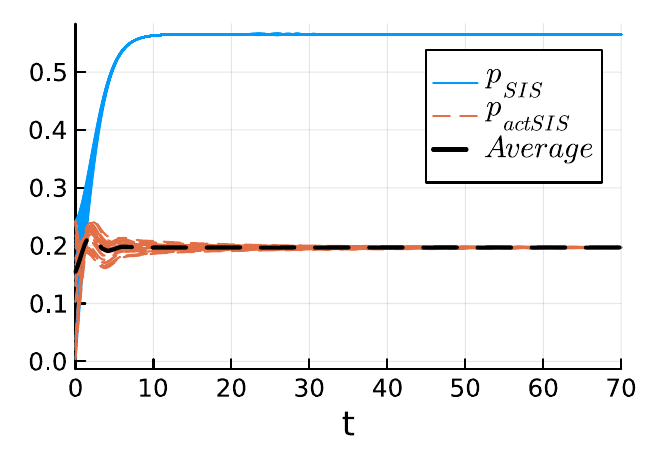}
        \caption{$\bar \beta = 0.23$}
    \end{subfigure}%
    \begin{subfigure}{.24\textwidth}
        \includegraphics[width=1.0\linewidth]{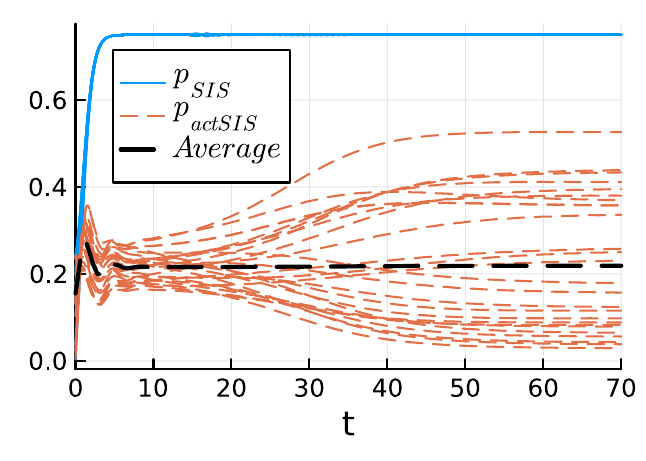}
        \caption{$\bar \beta = 0.4$}
    \end{subfigure}
    \caption{ 
    Trajectories of network SIS and network actSIS dynamics over $26$-node regular graphs with $d = 5$ and $\hat{d} = 21$. $\delta = 0.5, \mu = 0.2, \nu=8$ and $\tau_s=1$. Created using \cite{rackauckas2017differentialequations}.
    \label{fig:EEbifurcation_beta} }
\end{figure}

\section{Numerical explorations and Future Work} 
In Fig. \ref{fig:mixed-strategy} we simulate the network actSIS model with a mixed-strategy network of $10$ agents: $5$ of them are risk-averters and  $5$ are risk-ignorers (i.e. $\alpha_j(p_s) = 1$ for all $p_s$) as shown on the right. The simulation (left) shows that the infection level at the EE of the risk-averters is below the infection level of the risk-ignorers at the EE. This suggests that  even when in contact with risk-ignorers, risk averters are able to maintain relatively low infection levels.
\begin{figure}
    \centering
    \includegraphics[width = 0.48\textwidth]{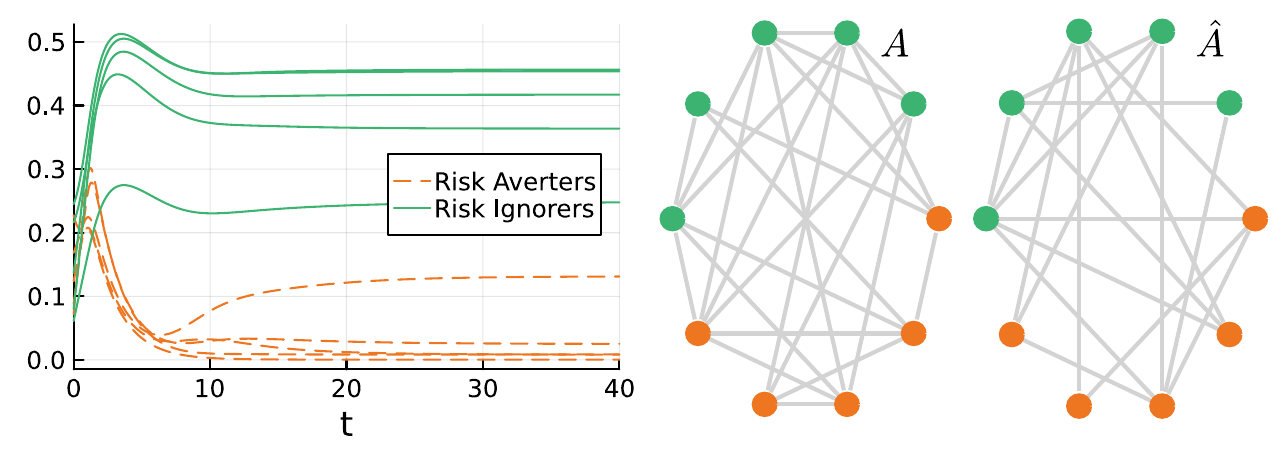}
    \caption{Simulations using \cite{rackauckas2017differentialequations} for a network with $10$ nodes. 
    $\bar\beta = 0.3, \delta = 0.5, \tau_s = 1, \mu = 0.2$, $\nu = 8$. 
    }
    \label{fig:mixed-strategy}
\end{figure}
We also explore the result proved in Theorem \ref{thm:uniform_EE_lower} for general contact networks. Fig. \ref{fig:no_regularity} suggests that, even for irregular contact matrices, the EE is reduced by the risk aversion strategy. We leave proving this 
result to future work. 
Another direction of future work is exploring an alternative hypothesis on risk perception, in which populations adapt their contact rates in response to the perceived rate of change in global infection level $\dot{p}_s$, rather than its net value $p_s$.

\begin{figure}
    \centering
    \includegraphics[width=0.48\textwidth]{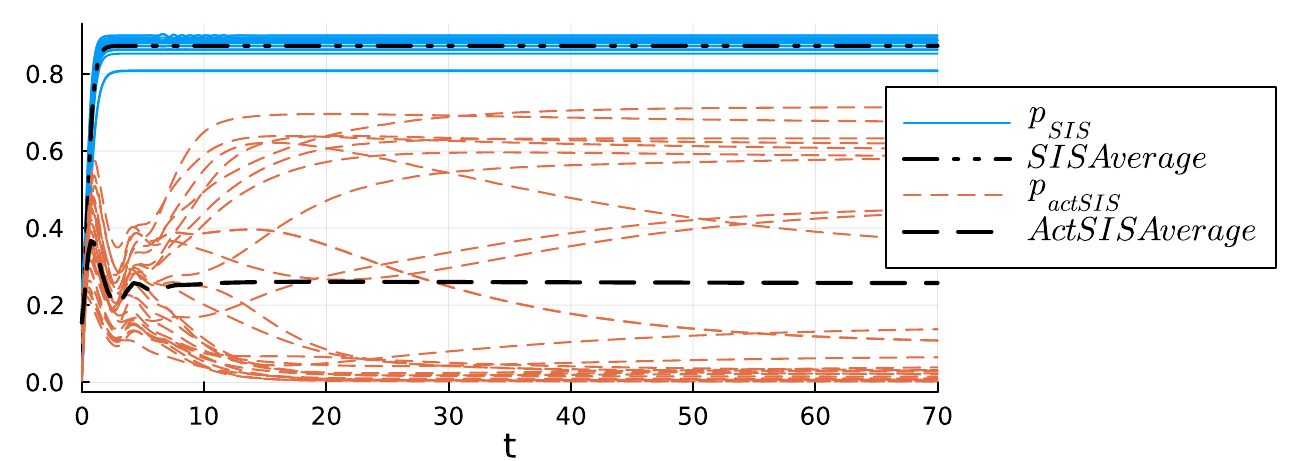}
    \caption{Simulations of \eqref{eq:NetworkSIS} and \eqref{eq:pdot},\eqref{eq:phat} over $26$ nodes. 
    $A$ and $\hat{A}$ are symmetric and generated randomly but not necessarily regular. 
    $\delta = 0.5, \bar\beta = 0.3, \tau_s = 1, \mu = 0.2, \nu = 8$. Created using \cite{rackauckas2017differentialequations}.}
    \label{fig:no_regularity}
\end{figure}



\section{Discussion}
We presented and analyzed the network actSIS model. 
We proved the onset of a transcritical bifurcation in the model. For regular contact and communication network graphs, we showed through proof and numerical simulation that the resulting EE is uniform and has lower infection levels than the EE of the 
network SIS model when populations use a risk aversion strategy. We proved for regular graphs that this equilibrium can become unstable in a bifurcation when the communication graph is sufficiently sparse. Simulations show that at this second bifurcation, an EE with nonuniform levels of infection emerges, despite the homogeneity in the system. Future work will expand these analyses to general graphs, 
and to strategies beyond risk aversion. 

\section*{Appendix}
\subsection*{Lyapunov-Schmidt Reduction for Theorem \ref{thm:bif}}
At the singular point, we compute the leading order coefficients of the Lyapunov-Schmidt reduction, following \cite[Chapter I, \S 3]{Golubitsky1985}: Define $\mathbf{x} = (\mathbf{p},\mathbf{p}_s)$ and let $F(\mathbf{x}, \bar \beta) = \big(f(\mathbf{p}, {\mathbf{p}}_s, \bar \beta), g(\mathbf{p}, {\mathbf{p}}_s, \bar \beta)\big)$, where $f,g$ are the right-hand side of \eqref{eq:pdot},\eqref{eq:phat}. Define the second order partial derivative of $F$ along vectors $\mathbf{y},\mathbf{z}$ and evaluated at some $(\mathbf{x}^*,\bar{\beta}^*)$ as $(d^2F)_{\mathbf{x}^*,\bar{\beta}^*}(\mathbf{y},\mathbf{z}) = \sum_{i=1}^{2N} \sum_{j = 1}^{2N} \frac{\partial^2 F}{\partial x_i \dots \partial x_j}(\mathbf{x}^*, \bar{\beta}^*) y_i z_j$.
A first-order partial directional derivative is defined analogously. Note that at the bifurcation point $(\mathbf{x},\bar\beta^*)$, the right and left eigenvectors of the Jacobian \eqref{eq:JIFE} are $\bar{\mathbf{v}} = (\mathbf{v^*}, \hat{\Delta}^{-1} \hat{A} \mathbf{v}^*)$ and $\bar{\mathbf{w}} = (\mathbf{w}^*, \mathbf{0})$. The quadratic LS coefficient is $h_{xx} = \langle \bar{\mathbf{w}}, d^2F_{0,\bar\beta^*} (\bar{\mathbf{v}}, \bar{\mathbf{v}}) \rangle =\sum_{j=1}^N - w^*_j  \bar\beta^*  v^*_j  \left(\sum_{m=1}^N  \alpha_j(0) a_{jm} v^*_m + \sum_{l=1}^N \alpha_j(0) a_{jl}  v^*_l \right)+ \sum_{j=1}^N w^*_j \sum_{l=1}^N \bar\beta^* a_j'(0) a_{jl} v_l^* (\hat{\Delta}^{-1} \hat{A} \mathbf{v}^*)_j= \bar{\beta}^* \left( \sum_{i = 1}^N w_i^*\left( (\hat{\Delta}^{-1} \hat{A} \mathbf{v}^*)_i (\alpha'(0) A \mathbf{v}^*)_i - 2 v_i^* (\alpha(0) A\mathbf{v}^*)_i \right) \right)$.
Defining $\beta = \bar{\beta} - \bar{\beta}^*$ and $\tilde{A} = \operatorname{diag}\{(\alpha_1(0),\dots,\alpha_N(0)\} A$ gives $h_{x \beta} = \langle \bar{\mathbf{w}}, d \left( \frac{\partial F}{ \partial \beta} \right)_{0,\bar\beta^*} (\bar{\mathbf{v}})\rangle=  \langle \mathbf{w}^*, \tilde{A} \mathbf{v}^* \rangle = \sum_{i = 1}^N w_i (\Tilde{A}\mathbf{v}_i)^*$, so $h_{x\beta}>0$.
\bibliographystyle{ieeetr}
\bibliography{references_ab}

\end{document}